%% file: main.tex
\documentclass[12pt]{article}
\usepackage{fullpage}
\usepackage{graphicx,amssymb,amsmath,amsthm}

\input{macros}
\input{environments}

\title{Sketching Persistence Diagrams}
\author{Donald R. Sheehy \and Siddharth Sheth}

\begin{document}
  \maketitle

  \input{abstract}
  \input{introduction}
  \input{background}
  \input{relatedwork}
  \input{hausdorff_to_bottleneck}
  \input{construction}
  \input{clarkson}
  \input{matching_update}
  \input{nbrhd_graphs}
  \input{conclusion}

  \bibliographystyle{abbrv}
  \bibliography{bibliography}

  \appendix
  \input{neighbor_analysis}

\end{document}

%% file: macros.tex
\newcommand{\dist}{\mathrm{d}}
\newcommand{\NN}{\mathrm{NN}}
\newcommand{\e}{\varepsilon}
\newcommand{\R}{\mathbb{R}}
\newcommand{\Z}{\mathbb{Z}}
\renewcommand{\flat}[1]{\underline{#1}}
\newcommand{\vor}{\mathrm{Vor}}
\newcommand{\rad}{\mathrm{rad}}
\newcommand{\segment}{\mathrm{seg}}
\newcommand{\nbrhd}{\mathrm{Nbrhd}}
\newcommand{\binbrhd}{\mathrm{BiNbrhd}}
\newcommand{\mult}{\mathrm{mult}}
\newcommand{\spread}{\Delta}
\newcommand{\opt}{\mathrm{Opt}}
\newcommand{\cost}{\mathrm{cost}}

%% file: environments.tex
\newtheorem{theorem}{Theorem}[section]
\newtheorem{lemma}[theorem]{Lemma}
\newtheorem{corollary}[theorem]{Corollary}

%% file: abstract.tex
\begin{abstract}
  Given a persistence diagram with $n$ points, we give an algorithm that produces a sequence of $n$ persistence diagrams converging in bottleneck distance to the input diagram, the $i$th of which has $i$ distinct (weighted) points and is a $2$-approximation to the closest persistence diagram with that many distinct points.
  For each approximation, we precompute the optimal matching between the $i$th and the $(i+1)$st.
  Perhaps surprisingly, the entire sequence of diagrams as well as the sequence of matchings can be represented in $O(n)$ space.
  The main approach is to use a variation of the greedy permutation of the persistence diagram to give good Hausdorff approximations and assign weights to these subsets.
  We give a new algorithm to efficiently compute this permutation, despite the high implicit dimension of points in a persistence diagram due to the effect of the diagonal.
  The sketches are also structured to permit fast (linear time) approximations to the Hausdorff distance between diagrams---a lower bound on the bottleneck distance.
  For approximating the bottleneck distance, sketches can also be used to compute a linear-size neighborhood graph directly, obviating the need for geometric data structures used in state-of-the-art methods for bottleneck computation.
\end{abstract}

%% file: introduction.tex
\section{Introduction}
\label{sec:introduction}

\emph{Persistent homology} (PH) is a topological invariant with a built-in metric.
Thus, qualitative shape information (topology) becomes quantitative (distances).
This is why PH is so useful as a meta-analysis tool; it can map an entire data set to a single point in a metric space, i.e., a \emph{persistence diagram} (PD).
The complexity of computing the distance between PDs is determined by the complexity of the PDs themselves, which are multisets of pairs of numbers.
The exact distance is computed by finding the minimum bottleneck matching between the sets.
Naturally, smaller diagrams lead to faster computation.

In this paper, we will explore methods for sketching PDs, producing much smaller diagrams while maintaining some guaranteed proximity to the original PD.
Given a PD $D$ with $n$ distinct (nondiagonal) points, we will produce a sequence of PDs $D_0,\ldots,D_n$ where each $D_i$ has $i$ distinct points.
Let $\e_i$ be the bottleneck distance $\dist_B(D,D_i)$ for all $i$.
The sequence has the property that $\e_0\ge \e_1 \ge \cdots \ge \e_n = 0$.
In other words, the sequence converges to $D$ in the bottleneck distance.
Moreover, the triangle inequality then gives a guarantee that for any PD $X$ and all $i$
\[
  |\dist_B(X,D) - \dist_B(X, D_i)| \le \e_i.
\]

In addition to computing the sequence of diagrams, we will also compute the optimal matching $M_i:D_i\to D_{i+1}$.
Thus, given a matching $M:X\to D_i$, we will be able to compute a matching $M_i\circ M:X\to D_{i+1}$ whose bottleneck cost is only increased by at most $\e_{i+1}$.
We will show how this can be used to efficiently reuse much of the work in computing a bottleneck distance to $D_i$ when computing the distance to the finer approximation $D_{i+1}$.
Surprisingly, the total space required to represent the sequence of diagrams as well as the sequence of matchings will only be $O(n)$.
So, for a constant factor extra space, the PD can be stored in a way that allows for fast approximate distance computations.

% If we did not allow multiplicity, the optimal approach would be to take the $k$ points of $D$ with the largest persistence (measured death time minus birth time).
% The desire to have few distinct points comes from the way fast algorithms for bottleneck matching exploit geometric search algorithms in the inner loop.
% Having many fewer points can dramatically improve the performance.

% The Approach:
Our approach is based on a greedy permutation (also known as a farthest-point traversal) of the points in the persistence diagram combined with a weighting scheme.
For a PD with $n$ points, we will produce $n$ different approximations, where the $i$th approximation has $i$ distinct weighted points.
During construction, we precompute both the bottleneck distance between the successive approximations, but also the matching that realizes this distance.
Thus, we have bounds on the error associated to any given approximation.

The main motivation for our work is to improve data analysis on spaces of persistence diagrams.
In many proximity search problems, it is often less important to compute distances than it is to guarantee that distances are larger than some bound in order to prune a search.
In simple metric spaces where distances are inexpensive to compute or are part of the input, this is usually accomplished by direct comparison.
For expensive to compute metrics such as the bottleneck distance, we would benefit from fast approximations.

%% file: background.tex
\section{Background}
\label{sec:background}

The extended plane is the set $\R_\infty^2 := (\R\cup \{\infty\})^2$.
Let $p = (p_b, p_d)$ be a point in $\R_\infty^2$.
The subscripts $b$ and $d$ on the coordinates refer to ``birth'' and ``death'' respectively.
The distance between points $p,q\in \R^2$ is defined using the $L_\infty$-norm
\[
  \|p-q\| = \max\{|p_b - q_b|, |p_d - q_d|\}.
\]
The distance between between points in $\R_\infty^2$ is defined similarly with the usual care required for arithmetic on $\infty$.

A persistence diagram (PD) is a multiset of points in $\R_\infty^2$ that contains the diagonal $\{(x,x) \mid x\in \R^\infty\}$ with infinite multiplicity.
For any multiset $X$, let $\flat{X}$ denote the underlying set, and for any $x\in \flat{X}$, let $\mult(x)$ be the multiplicity of $x$ in $X$.

Let $A$ and $B$ be PDs.
A bijection $M:A\to B$ is called a \emph{matching}.
The \emph{bottleneck cost} of $M$ is $\max_{a\in A} \|a - M(a)\|$.
The \emph{bottleneck distance} between PDs $A$ and $B$ is defined to be the minimum bottleneck cost over all matchings $M:A\to B$.

Every matching $M:A\to B$ induces a \emph{transportation plan}.
This is a function $T:\flat{A}\times \flat{B}\to\Z$ that counts the number of edges between a point $a\in \flat{A}$ and $b\in \flat{B}$.
More generally a function $T:\flat{A}\times \flat{B}\to\Z$ is a transportation plan between multisets $A$ and $B$ if for all $a\in \flat{A}$ and all $b\in \flat{B}$, we have
\[
  \sum_{b'\in \flat{B}}T(a,b') = \mult(a) \text{ and } \sum_{a'\in \flat{A}}T(a',b) = \mult(b).
\]
The bottleneck cost of a matching is easily computed from the transportation plan.
For this reason, it is not uncommon to see the bottleneck distance presented in terms of either.
Both matchings and transportation plans will be useful in this paper.
Matchings have the advantage that their composition is canonically defined.
Transportation plans have the advantage that they are simpler to represent and therefore reduce space usage.
Every transportation plan represents an equivalence class of matchings.

The bottleneck distance is a special case of the Wasserstein distance.
For a given matching $M$, the $p$-Wasserstein cost is
\[
  \cost_p(M) = \left(\sum_{x\in X} \dist(x, M(p))^p \right)^{1/p}.
\]
The bottleneck distance is the case of $p=\infty$.

\subsection{Greedy Permutations}

Given two subsets $A$ and $B$ in a metric space, the \emph{Hausdorff distance} between $A$ and $B$ is defined to be
\[
  \dist_H(A, B) := \max \{\sup_{a\in A}\inf_{b\in B} \dist(a,b), \sup_{b\in B}\inf_{a\in A} \dist(a,b)\}.
\]

Let $(X, \dist)$ be any metric space and let $P\subseteq X$ be finite.
Let $P$ be ordered $(p_0,\ldots, p_{n-1})$.
The $i$th \emph{prefix} of the ordering is denoted $P_i =\{p_0,\ldots, p_{i-1}\}$.
The ordering is a \emph{greedy permutation} if for all $i\in \{1,\ldots, n-1\}$,
\[
  \min_{q\in P_i}\dist(p_i, q) = \dist_H(P_i, P).
\]
In other words, each point $p_i$ is the farthest point from the prefix $P_i$.
For this reason, greedy permutations are also known as farthest point samples.
The point $p_0$ may be chosen arbitrarily.

For each $p_i$, the distance $\e_i = \dist(p_i, P_i)$ is called the \emph{insertion radius}.
By convention $\e_0 = \infty$.
By construction, for a greedy permutation, we have
\[
  \e_0\ge \e_1\ge \cdots \ge \e_{n-1}.
\]

The \emph{spread} $\spread$ of a point set is the ratio of the largest to smallest pairwise distances.
If the points are arranged in a greedy permutation, then the spread is at most $\frac{2\e_1}{\e_n}$.
We define the spread for a persistence diagram similarly, except that we ignore multiplicity and treat the entire diagonal as a single point.

%% file: relatedwork.tex
\section{Related Work}
\label{sec:relatedwork}

The state-of-the-art for computing the bottleneck distance between persistence diagrams is the work of Kerber et al.~\cite{kerber17geometry}.
They borrow the geometric insight from the work of Efrat et al.~\cite{efrat01geometry}, in which the Hopcroft-Karp matching algorithm~\cite{hopcroft73algorithm} is combined with geometric data structures to avoid constructing the entire bipartite graph.
% They incorporate geometric information not as yet included in Hopcroft and Karp's original computation in the form of a data structure that enables querying the nearest neighbors of a point within a certain distance.
With this approach Efrat et al.\ reduce the execution time of the matching algorithm from $O(n^{2.5})$ to $O(n^{1.5} \log n)$.
Kerber et al.\ use this idea to give a similar running time for computing the bottleneck distance persistence diagrams.

Kerber et al.\ also adapt Bertsekas' auction algorithm~\cite{bertsekas79distributed} to find an approximate Wasserstein distance between diagrams, including the adaptation by Bertsekas and Castanon~\cite{bertsekas89auction} that works with multiplicity.
% How they did it?
%  Adapted Efrat et al
%  which used HK plus geomtric data structures
% Also: experimented with Auction algo

% Cite dionysus? it includes Hera, the implementation of Kerber et al.

% Hopcroft and Karp\cite{hopcroft73algorithm} give an $O(n^{2.5})$ algorithm to perform a maximum bipartite matching between the vertices of a graph.\\
% Efrat et al.~\cite{efrat01geometry} improve on Hopcroft and Karp's algorithm by incorporating geometric information not as yet included in the computation.
% They use a data structure that enables querying the nearest neighbors of a point within a certain distance and execute that algorithm in time $O(n^{1.5} \log n)$.
% Kerber et al.~\cite{kerber17geometry} apply Efrat's geometric insight to compare PDs.
% % Multiplicities work with the Auction algorithm.
% Approximating a persistence diagram by using fewer points, but adding multiplicity is a natural way to speed up these algorithms and is the starting point for our work.

Our approach is based on greedy permutations of the points as originally presented for clustering (see~\cite{gonzalez85clustering,dyer85simple}).
An efficient algorithm to compute such permutations comes from the work of  Clarkson~\cite{clarkson03nearest} and his data structures for nearest neighbor search.
Har-Peled and Mendel~\cite{harpeled06fast} showed that Clarkson's algorithm runs in $2^{O(d)}n\log(\spread)$ time and $O(2^{O(d)}n)$ space, where $d$ is the doubling dimension and $\spread$ is the spread.
Our work extends Clarkson's algorithm and uses it to compute greedy permutations of PDs.

%Approximate Wasserstein matchings
There are many novel applications of the bottleneck and Wasserstein distance,  or its approximation.
Fasy et al.~\cite{fasy18approximate} consider the problem of approximating nearest neighbors of PDs.
Mumey~\cite{mumey18indexing} provides an approach to approximate nearest bottleneck distance queries over a finite point set by indexing a set of points to create a database and a \emph{trie}-based data structure.
Soler et al.~\cite{soler18lifted} track topological changes in time by finding matchings in a lifted Wasserstein metric.
Vidal et al.~\cite{vidal19progressive} provide an iterative algorithm to calculate progressively accurate approximations of the Wasserstein barycenters on a space of PDs.
All of these approaches could benefit from faster distance computation.

There are several previous approaches that transforming PDs into another representation to make certain tasks computationally simpler.
Bubenik~\cite{bubenik15statistical} introduces one such form called persistence landscape which enables statistical inference via standard statistical tests.
Adams et al.~\cite{adams17persistent} show that PDs can be converted to a vector representation called a persistence image.
Divol and Lacombe~\cite{divol20understanding} give a framework to study PDs by converting them to partial optimal transport problems and expressing them as Radon measures on the upper half plane.
Lacombe et al.~\cite{lacombe18large} describe a scalable framework to compute standard properties of PDs by reformulating them as optimal transport problems.

Additionally, there has been significant work in representing PDs as vectors leading to combination of machine learning theory with topological data analysis.
Carrière et al.~\cite{carriere15stable} define a stable topological point signature for shapes by extracting vectors from PDs.
Many machine learning techniques require the underlying space to be a Hilbert space.
% Thus, these techniques cannot be performed directly on the space formed by PDs coupled with the Wasserstein distance because this space is only a metric space.
It is shown by Reininghaus~\cite{reininghaus15stable} that it is possible to use topological information encapsulated in PDs in all kernel-based machine learning methods.
The kernel thus defined, however, performs poorly when encumbered with a large number of training vectors.
Zeppelzauer \cite{zeppelzauer16topological} attempts to overcome these limitations by analysing 3D surfaces using persistence image descriptors of Adams et al.~\cite{adams17persistent}.
Carrière et al.~\cite{carriere17sliced} implement a provably stable sliced Wasserstein kernel and an algorithm to approximate it.

%% file: hausdorff_to_bottleneck.tex
\section{From Hausdorff to Bottleneck Approximations}
\label{sec:hausdorff_to_bottleneck}

The main idea of this paper is to use greedy permutations to give approximations to a persistence diagram.
The greedy permutation is often used to give approximations that are close in Hausdorff distance.
If $A$ and $B$ are PDs, then it is possible that $\dist_H(\flat{A},\flat{B}) = 0$ and $\dist_B(A, B)$ is large.
Therefore, one should not, in general, expect that a good Hausdorff approximation will give a good bottleneck approximation.
The same holds for other Wasserstein metrics.

There is one important case in which we \emph{can} relate the Hausdorff distance and the Bottleneck distance---when one diagram is a subset of the other and the multiplicities of its points have been carefully adjusted.
This section gives the construction and proves the equivalence of the Hausdorff and Bottleneck distances for that case.

\paragraph*{Natural Reweighting}

A \emph{reweighting} of a PD $A$ is a new persistence diagram $A'$ formed by assigning new multiplicities to the points.
Thus, if $A$ is a reweighting of $A$, then $\flat{A} = \flat{A'}$.

Given $A\subseteq B$, the \emph{natural reweighting} of $A$ with respect to $B$ is one that assigns a multiplicity to each point $a$ in $\flat{A}$ according to the number of points in $B$ having $a$ as their nearest neighbor.
That is, for $a\in \flat{A}$, we have $\mult(a)$ is the number of points of $B$ that are closer to $a$ than to any other point of $\flat{A}$.
Ties can be broken arbitrarily, possibly resulting in non-uniqueness.

\begin{lemma}\label{lem:natural_reweighting}
  Let $A$ and $B$ be PDs with $A\subseteq B$.
  If $A'$ is a natural reweighting of $A$, then
  \[
    \dist_B(A', B) = \dist_H(\flat{A}, \flat{B}).
  \]
\end{lemma}
\begin{proof}
  All points in $A\cap B$ will be matched to themselves.
  The points of $B\setminus A$ will be matched to their nearest neighbor.
  This exactly matches each point of $A'$ with the correct multiplicity.
  This is a matching whose bottleneck is $\max_{b\in B}\dist(b,A) = \dist_H(A,B)$.
  This matching must be optimal, because every edge from $b\in B$ is the globally shortest possible to a point in $A$.
\end{proof}

\paragraph*{Optimal Transport}

The natural reweighting can be understood in terms of optimal transport.
Specifically, for $A\subseteq B$, the natural reweighting $A'$ corresponds to the transportation plan that moves every point $b$ of $B$ to its nearest point $\NN_{\flat{A}}(b)$ in $\flat{A}$.
Formally, the transportation plan is
\[
  T(a,b) := \begin{cases}
    \mult(b) & \text{if } a = \NN_{\flat{A}}(b)\\
    0 & \text{otherwise.}
  \end{cases}
\]

The transportation plan $T$ is optimal for any Wasserstein metric.
Any other transportation plan would suffer increased cost for each point of $B$ that is not moved to its nearest neighbor.
% TODO: Maybe delete the following lemma
% From this it follows that $W_p(A', B)$ can be computed directly using the following lemma.

% \begin{lemma}\label{lem:optimal_transport}
%   Let $A$ and $B$ be PDs such that $A\subseteq B$.
%   If $A'$ is the natural reweighting of $A$, then
%   \[
%     W_p(A', B) = \left(\sum_{b\in B}\dist(b,A')^p\mult(b)\right)^{\frac{1}{p}}.
%   \]
% \end{lemma}

%% file: construction.tex
\section{Greedy Permutations of PDs}
\label{sec:construction}

Given a PD, the diagonal and the multiplicity of points makes it impossible to give a greedy permutation directly.
With a slight adjustment, we can define a greedy permutation of the nondiagonal points of a PD.
Let $D$ be a PD and let the nondiagonal points of $\flat{D}$ be ordered $p_0,\ldots, p_{n-1}$.
The $i$th approximate diagram $D_i$ is the natural reweighting of the $i$th prefix of the ordering with the diagonal added to make it a PD.
So, $D_0$ is the empty diagram consisting of just the diagonal and $D_n$ is $D$.
The ordering is greedy if for all $i \in \{0,\ldots,n-1\}$,
\[
  \min_{q\in D_i} \|p_i - q\| = \dist_H(\flat{D_i}, \flat{D}).
\]
The sequence $(D_0,\ldots, D_n)$ is called a \emph{greedy PD sketch} of $D$.

By always starting with the diagonal, the choice of $p_0$ is not arbitrary as is the case of greedy permutations in other metrics.
The permutation will always start with $p_0$ as a point of maximum persistence.
Also, it is not relevant to consider multiplicities when finding greedy permutations of PDs.
Once all the distinct points have been added, the Hausdorff distance will be zero.

Because $D_i$ is the natural reweighting with respect to $D$, the result is a sequence of bottleneck approximations to $D$.
Up to a factor of two, these approximations are optimal for their size in the following sense.

\begin{theorem}
  Let $D$ be a persistence diagram and let $(D_0,\ldots, D_n)$ be a greedy PD sketch of $D$.
  For all $i$, let $\opt_i$ be the PD with $i$ distinct points that minimizes $\dist_B(D, \opt_i)$.
  Then, for all $i$, the approximation $D_i$ satisfies
  \[
    \dist_B(D, D_i) \le 2 \dist_B(D, \opt_i).
  \]
\end{theorem}

\begin{proof}
  The proof follows the same pattern as the standard proof that greedy permutations yield $2$-approximate $k$-centers in metric spaces~\cite{gonzalez85clustering,dyer85simple}.
  Let $r = \dist_B(D, \opt_i)$.
  Any point of $D$ paired with the diagonal in the optimal matching with $\opt_i$ can also be matched with the diagonal in the matching with $D_i$.
  So, it will suffice to consider points of $D$ matched to nondiagonal points.
  By construction, the distance between any pair of nondiagonal points in $D_i$ is at least $\min_{q\in D_i}\|p_i - q\| = \dist_H(\flat{D_i}, \flat{D})$.
  So, if any two points of $D_i$ are matched to the same nondiagonal point $q$ in $\opt_i$, then, by the triangle inequality,
  \[
    \dist_H(D_i, D) \le 2r.
  \]
  So, we may assume that the bottleneck matching that realizes $\dist_B(D, \opt_i)$ matches each point of $D_i$ with a unique point of $\opt_i$.
  Using the triangle inequality, every point of $D$ is within $2r$ of a point of $\opt_i$.
  Therefore, we have shown that $\dist_H(\flat{D_i}, \flat{D})\le 2r$, and so, by Lemma~\ref{lem:natural_reweighting}, it follows that $\dist_B(D_i, D) \le 2r$.
\end{proof}

The preceding theorem shows the sense in which the PD sketch is approximately optimal for its size.
When used as a proxy for a full diagram, there is a clear upper bound on the error; using $D_i$ instead of $D$ introduces at most $\e_i$ error as shown in the following lemma.

\begin{lemma}\label{lem:error_bound}
  Let $X$ and $D$ be persistence diagrams and let $i$ be a nonnegative integer.
  Let $D_i$ be the $i$th PD in a greedy PD sketch of $D$ and let $\e_i=\dist_H(\flat{D_i}, \flat{D})$.
  Then,
  \[
    |\dist_B(X,D) - \dist_B(X, D_i)| \le \e_i.
  \]
\end{lemma}
\begin{proof}
  Because $D_i$ is the natural reweighting of a subset of $D$, Lemma~\ref{lem:natural_reweighting} implies that $\dist_B(D, D_i) = \dist_H(\flat{D}, \flat{D_i}) = \e_i$.
  The desired inequality, then follows from the triangle inequality for the bottleneck distance.
\end{proof}

\subsection{Transportation Plans}

In addition to the approximate PDs, we also want to compute a matching or a transportation plan that take us from $D_i$ to $D_{i+1}$.
The difference between these PDs is
\begin{itemize}
  \item the addition of point $p_i$,
  \item some mass moves from points in $D_i$ to $p_i$, and
  \item some mass moves from the diagonal to $p_i$.
\end{itemize}
By the definition of the natural reweighting, the movement of mass is just a count of how many points of $D$ have $p_i$ as their nearest neighbor.
By the triangle inequality in the plane, any such transportation plan must be optimal as the source and target of every unit of mass is known, the straight lines have minimal cost.

The natural reweighting of $D_{i+1}$ is with respect to $D$ rather than $D_i$.
As a result, the bottleneck distance between $D_i$ and $D_{i+1}$ may be larger than the bottleneck distance between $D_i$ and $D$.
Specifically, we have
\[
  \e_i = \dist_H(\flat{D_i}, \flat{D_{i+1}}) \le \dist_B(D_i, D_{i+1}) \le \e_i + \e_{i+1}.
\]

The PD sketch for a PD with $n$ points contains $n+1$ PDs as well as $n$ transportation plans.
With a little care, the entire size to represent all of this is still only $O(n)$ as shown in the following theorem.

\begin{theorem}
  Given a persistence diagram $D$ with $n$ points, the greedy PD sketch and the optimal transportation plans can be represented in $O(n)$ space.
\end{theorem}
\begin{proof}
  Without the multiplicities, the sequence of diagrams is represented by the greedy permutation of the points.
  The multiplicities are encoded in the transportation plans from $D_i$ to $D_{i+1}$.
  In such a transportation plan, all the mass that moves is shifted to the newly added point $p_i$.
  Say that a point $q\in D_i$ is a neighbor of $p_i$ if it is one of the points that shifts some of its mass to $p_i$.
  That is, there is some point $x\in D$ such that $NN_{D_i}(x) = q$ and $NN_{D_{i+1}}(x) = p_i$.
  By the greedy ordering,
  \[
    \|x- p_j\| < \|x- q\| \le \e_i.
  \]
  Moreover, the greedy permutation guarantees that the neighbors of $p_i$ are all $\e_i$-separated.
  So, the squares of side-length $\e_i$ centered at the neighbors of $p_i$ are all disjoint and contained in the square of side length $5\e_i$ centered at $p_i$.
  It follows that there are at most $25$ neighbors including the diagonal.
  Thus, the $i$th transportation plan can be represented by a list of at most $25$ neighbors along with an integer for each counting the amount of mass moved.
  In total this is $O(n)$ numbers to store.
\end{proof}

%% file: clarkson.tex
\section{Modifying Clarkson Algorithm for Greedy Permutations of Persistence Diagrams.}
\label{sec:clarkson}

It is possible to compute the greedy permutation of a persistence diagram using the standard quadratic time algorithm~\cite{gonzalez85clustering,dyer85simple}.
The naive approach is simply to treat the entire diagonal as the first point in the permutation.
To get a faster algorithm, a simple and effective approach due to Clarkson can compute a greedy permutation in $O(n\log \spread)$ time in low-dimensional metrics.
However, the inclusion of the diagonal in a persistence diagram is an obstacle to direct application of Clarkson's algorithm.
The reason is that treating the diagonal as a point breaks the triangle inequality, e.g., consider two points that are both close to the diagonal but far from each other (see Fig.~\ref{fig:triangle_ineq_fail}).
If one enforced the triangle inequality, the impact would be that the diagram could appear to have high doubling dimension; all $n$ points could be one unit from the diagonal and thus equidistant (exactly two units) from each other.
In this section, we will show how to modify the algorithm so that we can achieve the same $O(n \log \spread)$ running time for persistence diagrams.
The main insight is to augment the PD with its projection to the diagonal and modify the way distances are computed.
We will also give an example where the direct application of Clarkson's algorithm devolves to quadratic time.

\begin{figure}
  \centering
  \includegraphics[width=0.7\textwidth]{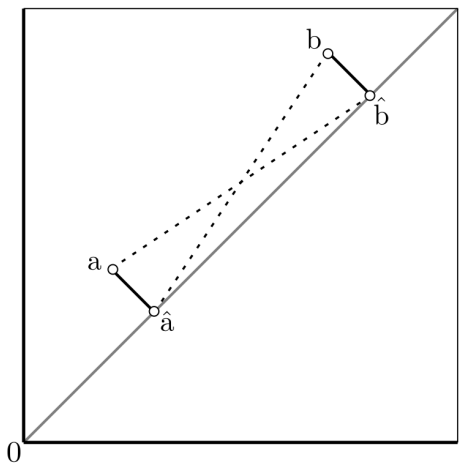}
  \caption{The triangle inequality fails when treating the diagonal as a point. If $\dist(a,b)\gg \dist(a,\hat{a})$ and $\dist(a,b) \gg \dist(b, \hat{b})$ then the triangle inequality fails as $\dist(a,b) > \dist(a,\hat{a}) + \dist(\hat{a},b)$}
  \label{fig:triangle_ineq_fail}
\end{figure}

\subsection{Overview of Clarkson's Algorithm}

In his work on nearest neighbor search in metric spaces~\cite{clarkson97nearest, clarkson99nearest,clarkson03nearest,clarkson06nearest}, Clarkson developed several data structures based on a kind of discrete Voronoi diagram.
Points are inserted into the structure one at a time and the uninserted points are assigned to their nearest neighbors among the inserted points.
The set of points whose nearest neighbor is a point $p$ are called the \emph{(discrete) Voronoi cell} of $p$ and are denoted $\vor(p)$.
A neighborhood graph is constructed on the inserted points with edges between points $p$ and $q$ representing the possibility that inserting a point in $\vor(p)$ would alter $\vor(q)$ or vice versa.
As points are inserted, updates to the structure remain local with respect to this neighborhood graph.
Some extra edges are maintained to simplify construction, pruning only those between points whose distance is more than some constant times their radii.

The Voronoi cells provide exactly the information required to compute both the natural reweighting as well as the optimal transportation plans.
The improvements in running time come from the sparsity of the neighborhood graph.
That sparsity will also translate into the sparsity of the PD sketch.

The simplest version of Clarkson's data structure is called $sb$ and was implemented in C.
It was originally designed to use a random permutation of the input points, but after experimental analysis, it was shown that performance improved when using a greedy permutation~\cite{clarkson03nearest}.
The data structure also finds the greedy permutation efficiently as part of the construction: knowing the Voronoi cells of the points added so far makes it easy to add the next farthest point by storing the Voronoi cells in a priority queue.
% When we refer to Clarkson's greedy permutation algorithm (or simply Clarkson's algorithm), we are referring to this variant of the $sb$ construction algorithm.
Later, Har-Peled and Medel showed that this algorithm runs in $O(n\log \spread)$ time in doubling metrics.
A Python implementation of this algorithm is also available~\cite{sheehy20greedypermutations}.

Consider a PD in which all the points are on a line one unit from the diagonal and are spaced out two units apart.
With a slight perturbation, any given ordering can be the greedy permutation.
If the ordering is sorted along the line, then each new point would require checking all the other points to see if they move.
Thus the total running time is quadratic, even though the spread is linear. This is illustrated in Fig. ~\ref{fig:quad_case}.
The reason this seems to violate the $O(n\log \spread)$ running time is that the big-O hides a term exponential in the doubling dimension.
If the distances are replaced with shortest path distances to enforce the triangle inequality, the doubling dimension of this example becomes $\log n$ because the the path from a point to the diagonal plus the diagonal to another point results in all points equidistant.
So, in order to compute the greedy permutation in subquadratic time using Clarkson's algorithm, one must treat the diagonal differently.

\begin{figure}
    \includegraphics[width=0.19\textwidth]{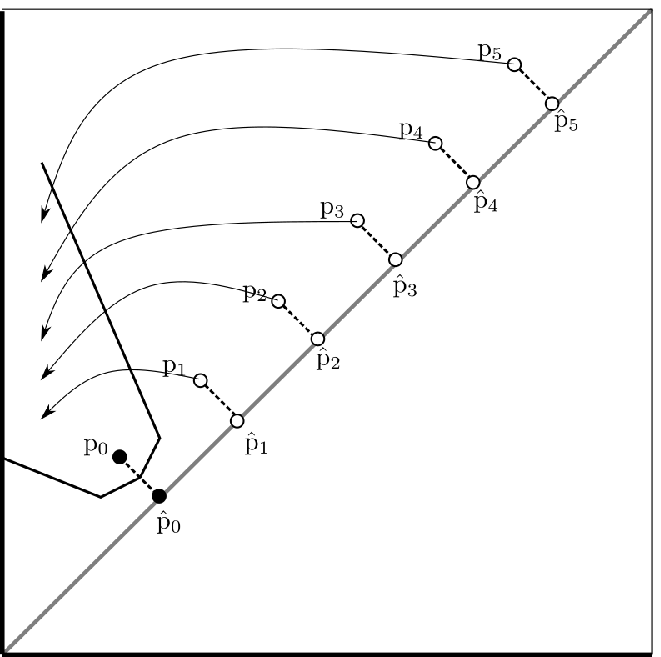}
    \includegraphics[width=0.19\textwidth]{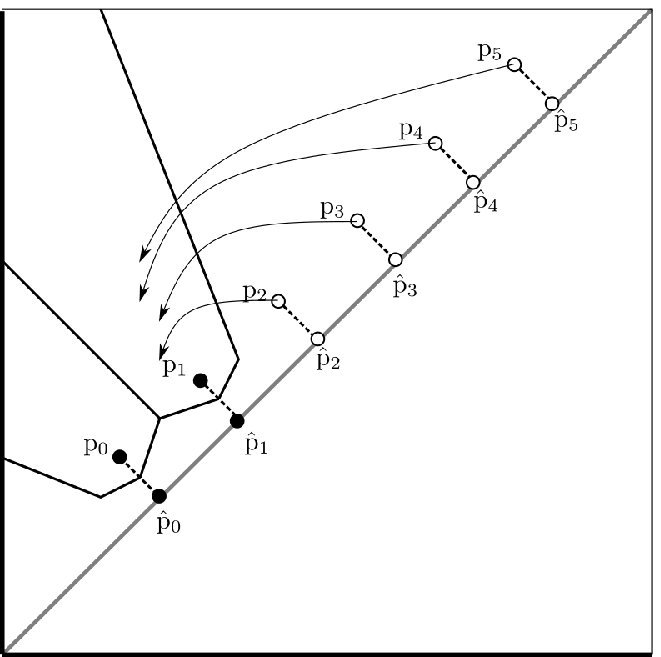}
    \includegraphics[width=0.19\textwidth]{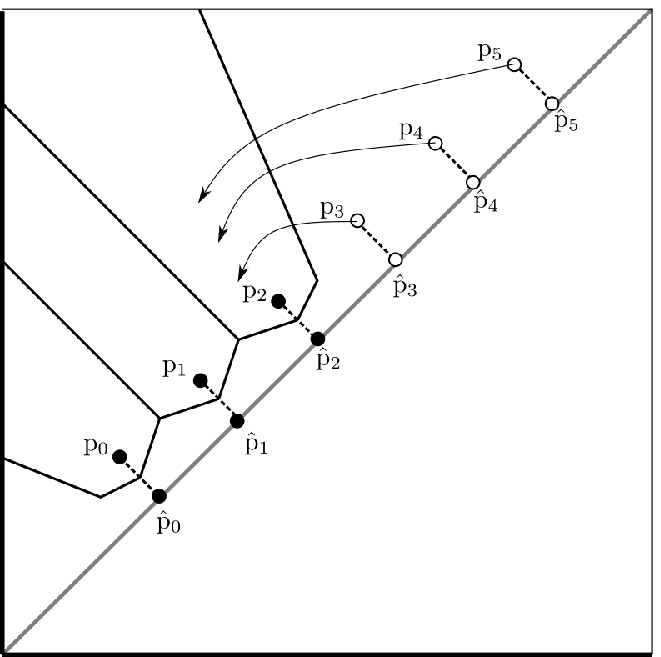}
    \includegraphics[width=0.19\textwidth]{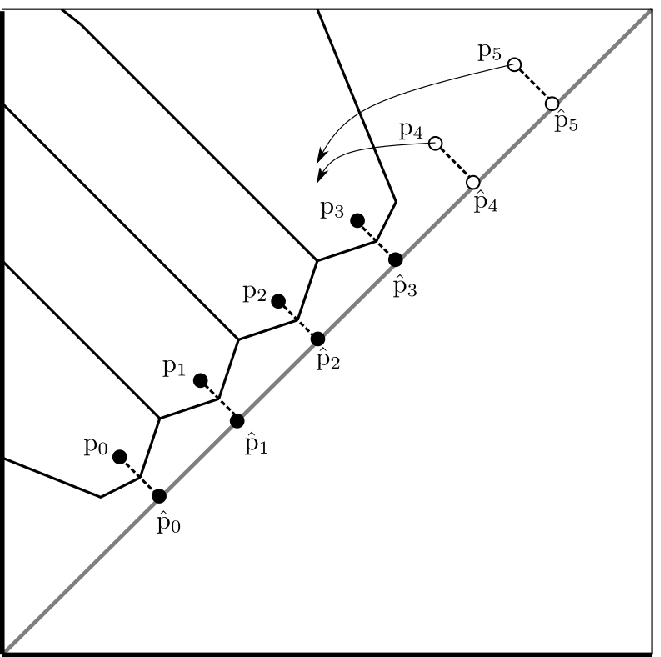}
    \includegraphics[width=0.19\textwidth]{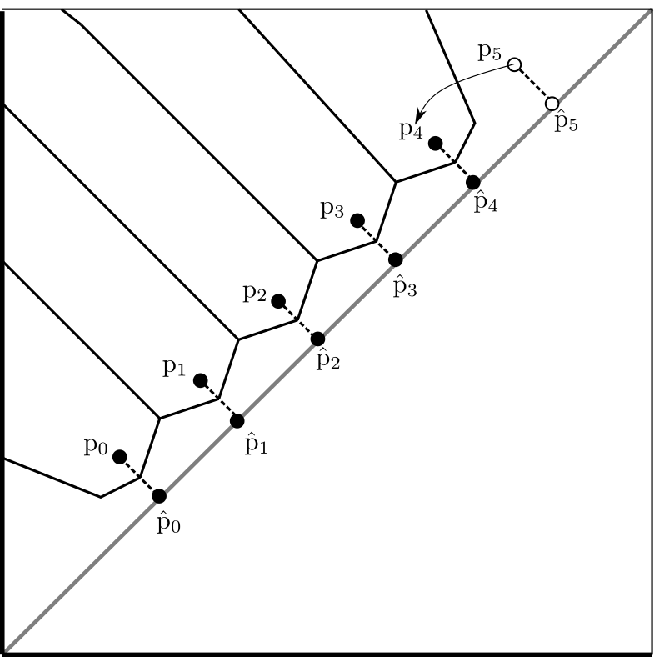}
  \caption{In this example, Clarkson's algorithm devolves to quadratic time. As points are added $p_0$ to $p_5$, all other points have to be checked to see if they move.}
  \label{fig:quad_case}
\end{figure}

\subsection{Using Projections}

Let $D$ be a persistence diagram.
Let $X$ be the nondiagonal points of $D$ and let $\hat{X}$ be the projections of the points of $X$ to the diagonal.
For each $a\in X$, we will write $\hat{a}$ for its projection.
Ultimately, we will compute the greedy permutation of $X\cup \hat{X}$ and retain only the ordering on $X$.

Let $S$ be a subset of $X\cup \hat{X}$.
The points $\hat{X}$ naturally partition the diagonal into segments where for any $\hat{x}\in S\cap \hat{X}$,
\[
  \segment_S(\hat{x}) := \{\hat{y} \mid \text{for all }\hat{z}\in S\cap \hat{X}, \|\hat{x} - \hat{y}\| \le \|\hat{z} - \hat{y}\|\}.
\]
When partitioning the points into discrete Voronoi cells, we will consider the Voronoi cells of the segments for the diagonal points.
Formally,
\[
  \dist_S(x,y) :=
    \begin{cases}
      \min_{\hat{x}\in \segment_S(x)} \|\hat{x} - y\| & \text{ if }x\in \hat{X}\\
      \|x - y\| & \text{otherwise.}
    \end{cases}
\]

This way of computing distances does not give a metric, but it will give the correct notion of discrete Voronoi cells, defined as
\[
  \vor_S(x) := \{y\in X\cup \hat{X} \mid \text{ for all }z\in S, \dist_S(x,y)\le \dist_S(z,y)\}.
\]

Note that the addition of more points of $\hat{X}$ into $S$ will not affect the Voronoi cells of points in $S\cap X$.
This is the desired invariant because it means that the union of the Voronoi cells of the diagonal points only depends on the nondiagonal points.
The points of $S\cap \hat{X}$ serve only to partition what would otherwise be one large cell associated to the diagonal.

The \emph{projection distance} is
\[
  r(x, y):= \begin{cases}
    \|y-\hat{y}\| & \text{if }x\in \hat{X} \text{ and } y\in X\\
    \|x-y\| & \text{otherwise.}
  \end{cases}
\]
The \emph{radius} of a Voronoi cell of a point $x\in S$ is
\[
  \rad(x) := \max_{y\in \vor_S(x)} r(x, y).
\]

The modified Clarkson algorithm then computes the greedy PD sketch by maintaining these Voronoi cells.
The next point added at each step is the farthest point (in projection distance) in the Voronoi cell of maximum radius.
The neighborhood graph connects two Voronoi cells as long as their distance is less than four times the larger of their radii.

The key to the analysis of Clarkson's aglorithm is that the neighborhood graph maintains constant degree.
One must keep enough edges so that the following two conditions hold.
\begin{enumerate}
  \item The Voronoi cell of a newly inserted point is contained in the union of the cells of the neighbors of its parent just prior to insertion.
  \item The neighbors of a newly inserted point are contained in the union of the neighbors of neighbors of its parent just prior to insertion.
\end{enumerate}

Our changes to how distances are computed produce only a small change to the analysis of Clarkson's algorithm, requiring us to store edges in the neighborhood graph up to four times the radius (rather three times the radius as in the original).
Appendix~\ref{sec:neighbor_analysis} gives the details on why this small change is sufficient (by repeated use of the triangle inequality).

\begin{theorem}\label{thm:contruction_correctness}
  The modified Clarkson algorithm restricted to the nondiagonal points gives a greedy PD sketch.
\end{theorem}
\begin{proof}
  It suffices to prove that each time a nondiagonal point $p_i$ is added, it is the farthest point from $D_i$.
  The algorithm explicitly chooses the farthest point at each step, so we need only observe that the inclusion of the projections do not change the Voronoi cells of the nondiagonal points.
  In other words, the projections do not change the order in which the nondiagonal points are added.
  This follows from the definition of the projection distance.
\end{proof}

\begin{theorem}\label{thm:running_time_clarkson}
  The greedy PD sketch can be computed in $O(n\log \spread)$ time.
\end{theorem}
\begin{proof}
  Using the projections, the analysis of the running time is the same as the standard analysis as given in Har-Peled and Mendel~\cite{harpeled06fast} and also in Clarkson~\cite{clarkson03nearest}.
  The key idea is to count the number of times any point is considered for moving into a new Voronoi cell.
  By a volume packing argument, this can happen only a constant number of times before the maxmimum radius goes down by at least a factor of $2$.
  The one caveat is that the projection could artificially increase the spread.
  This is resolved by stopping the algorithm as soon as all of the nondiagonal points are inserted.
  This happens at scale $\frac{1}{\spread}$ times the diameter and therefore, the total running time is $O(n\log \spread)$ as desired.
\end{proof}

A natural improvement in many instances would be to stop early and return a prefix of the PD sketch.
In particular, if one has other sources of error---such as the grid size in persistence images~\cite{adams17persistent}---one need not compute the full sketch.
For constant precision, the running time will be linear.

\begin{corollary}
  Let $D$ be a PD, and let $R_{\max}$ be the maximum persistence of any point in $D$.
  For a fixed precision $\e$, a partial greedy PD sketch $(D_0,\ldots, D_k)$ can be computed such that
  \[
    \dist_B(D_k, D)\le \e
  \]
  in $O(n \log \frac{R_{\max}}{\e})$.
\end{corollary}

%% file: matching_update.tex
\section{Updating a Matching}
\label{sec:matching_update}

Starting from a greedy PD sketch of $D$ and a second PD $X$, an approximation $D_i$ in the sketch can be used to estimate $\dist_B(X,D)$.
In doing so, an optimal bottleneck matching between $D_i$ and $X$ is computed.
Choosing a larger value of $i$ will increase the precision.
In this section, we show how to bootstrap the computation that has already been done for $D_i$ to get good matching between $X$ and $D_{i+1}$.
The principle idea is to compose a matching $X\to D_i$ with a matching $D_i\to D_{i+1}$ that is consistent with the transportation plan $T_i$.
The result will not necessarily be the optimal matching $X\to D_{i+1}$, but it will satisfy the cost guarantee of Lemma~\ref{lem:error_bound}.
Thus, in many cases, most of the work of finding the bottleneck matching will already be done.

\paragraph*{Naive Update}
The transportation plan $T_i: \flat{D_i} \times \flat{D_{i+1}} \to \Z$ encodes an equivalence class of matchings.
Any matching $M_i:D_i\to D_{i+1}$ that has $T_i(q,q')$ edges from $q$ to $q'$ is in this class.
One can easily choose such a matching arbitrarily and compose it with the previously computed optimal matching $M:X\to D_i$.
That is, $M_i\circ M:X\to D_{i+1}$ is a matching.
Because $M_i$ has bottleneck cost at most $\e_i$, the new matching will increase in cost by at most $\e_i$ (by the triangle inequality).

\paragraph*{A Short Auction}
It is possible to choose locally optimal matching update consistent with $T_i$.
This is perhaps most easily understood in terms of the $p$-Wasserstein cost of the matching.
Minimizing this cost is equivalent to minimizing the $p$th power of the cost.
So, replacing one edge $xq$ with an edge $xp_i$ in a matching $M'$ results in the following change.
\[
  \cost_p(M')^p - \cost_p(M)^p = \dist(x,p_i)^p - \dist(x,q)^p.
\]
Thus, for each such edge $xq$, we compute the corresponding $p$th power change in cost.
Then, we make the update the edges in order of this change.
This resembles an abbreviated version of the auction algorithm~\cite{bertsekas79distributed}.

For bottleneck matchings, one cannot assign costs using the $p$th power of the distances, because $p=\infty$ in that case.
Instead, one can find the optimal matching from the neighbors of $q$ in the matching $M$ to the points $q$ and $p_i$ (with multiplicities).
This only requires sorting the $O(n)$ edges by their length.
The bottleneck matching can be found in the minimum prefix of this ordering that has sufficiently many edges incident to the points of $X$ and the points $q$ and $p_i$ (by Hall's Theorem).
We call the resulting matching, \emph{the locally optimal matching consistent with $T_i$}.

In either case, the time is bounded by the cost of sorting the points of matched to $q$ for each $q$ that shifts some mass to $p_i$ in the transportation plan $T_i$.
We have already shown that there are only at most 25 such points $q$, but with multiplicity, there could be as many as $|X| = \Theta(n)$ edges to consider.
This makes the overall, worst-case running time $O(n\log n)$ for an update to the matching.

% TODO: is this even worth saying.
% For any updates such that $q$ and $p_i$ are very close (i.e., $\|q-p\| \le R + \e_i$ where $R$ is the current bottleneck radius), any choice of matching will result in the same bottleneck cost.

\paragraph*{Amortizing the Cost of Updates}
Although the worst case cost of updating the matching when going from $D_i$ to $D_{i+1}$ is $O(n\log n)$, not every such update can be that expensive.
Below, we show that a sequence of these updates from $D_i$ to $D_j$ such that $\e_i \le 2\e_j$ will also only take $O(n \log n)$ time.
This means adding points in the sketch to halve the error (double the precision) is asymptotically the same (in the worst-case) as updating the matching for one new point.

\begin{theorem}
  Let $(D_0,\ldots, D_n)$ be a greedy PD sketch of $D$.
  For any $i,k$ such that $\e_i \le 2\e_k$, the total cost of computing the locally optimal matchings $M_j$ consistent with $T_j$ for all $i\le j \le k$ can be computed in $O(n\log n)$ time.
\end{theorem}
\begin{proof}
  The total cost is $\sum_{q\in Q} \mult(q) \log(\mult(q))$ where $Q$ is the set of points whose mass changes at some point in the sequence of updates.
  Recall that by the definition of the natural reweighting used in the greedy PD sketch, the multiplicity of a point $q\in D$ is equal to the number of points in a discrete Voronoi cell at some point in the construction.
  If $T_j(q, p_j)> 0$, then $q$ and $p_j$ have neighboring Voronoi cells in some step of the construction.
  This means that every point in $\vor(q)$ is touched when $p_j$ is inserted to see if it must be move.
  So, as we observed in the analysis of the construction (Theorem~\ref{thm:running_time_clarkson}), any given point can be touched at most $O(1)$ times before the insertion radius decreases by a factor of $2$.
  As the edges incident to any point $q$ in the matching are in correspondence with the points of $\vor(q)$, it follows that each such edge participates in at most $O(1)$ updates.
  In other words, $\sum_{q\in Q} \vor(q) = \sum_{q\in Q} \mult(q) = O(n)$.
  So, the total cost is
  \[
    \sum_{q\in Q} \mult(q) \log(\mult(q)) \le \sum_{q\in Q} \mult(q) \log(n) =  O(n\log n).
  \]
\end{proof}

%% file: nbrhd_graphs.tex
\section{Filtered Neighborhood Graphs}

Matchings are computed on neighborhood graphs.
In this section, we show how the greedy sketch computation can also simplify the construction of these graphs.
We then show how this same construction allows for fast Hausdorff approximation to between PD sketches, leading to fast lower bounds on the bottleneck distance.
% Although modern methods use geometric data structures like k-d trees to avoid explicitly constructing a neighborhood graph,

The \emph{$\alpha$-neighborhood graph} on a set $P$ is the graph
\[
  \nbrhd(P, \alpha) := (P, \{(p_i, p_j) \mid \dist(p_i, p_j)\le \alpha\}).
\]
If the points of $P$ are all pairwise $\lambda$ apart, then the degree of any vertex in $\nbrhd(P, \gamma\lambda)$ cannot exceed $(2\gamma + 1)^2$.%\footnote{For other metrics, the degree bound will be $\gamma^{O(d)}$, where $d$ is the doubling dimension.}
This follows because the squares of side length $\lambda$ centered at the neighbors will be disjoint and contained in the square of side length $(2\gamma + 1)\lambda$.

Let $P = (p_0, \ldots, p_{n-1})$ be a set ordered according to a greedy permutation.
The \emph{$\gamma$-filtered neighborhood graph} on $P$ is the graph with vertices $P$ and edges $(p_j, p_i)$ whenever $i< j$ and $\dist(p_i, p_j) < \gamma\lambda_j$.
Because $P_j$ is a $\lambda_j$-net for all $j$, there will be at most $(2\gamma + 1)^2$ neighbors of $p_j$ that precede it in the greedy permutation.
Thus, the total number of edges is $(2\gamma + 1)^2n$.
Moreover, the graph contains $\nbrhd(P_i, \gamma\lambda_i)$ for all $i$.

When computing the greedy permutation, one can compute the filtered neighborhood graph at the same time in the same asymptotic running time.
This is the underlying idea in the Clarkson algorithm (the graph is an undirected version of the $sb$ data structure).

For a pair of compact sets $P,Q$, the bipartite $\alpha$-neighborhood graph is
\[
  \binbrhd(P, Q, \alpha) := (P\sqcup Q\mid \{(p,q)\in P\times Q \mid \dist(p,q)\le \alpha)
\]
The Hausdorff distance between $P$ and $Q$ is the minimum $\alpha$ such that $\binbrhd(P,Q,\alpha)$ contains no isolated vertices.
The bottleneck distance between $P$ and $Q$ is the minimum $\alpha$ such that $\binbrhd(P,Q,\alpha)$ contains a perfect matching.
If $P$ and $Q$ are persistence diagrams, one adds the diagonal as an extra vertex to each set and adds edges from points to the diagonal if their projection to the diagonal is within $\alpha$.

% Let $P$ and $Q$ be given their greedy orderings.
% Let $P_\lambda$ denote the $P_i$ such that $\lambda_i< \lambda \le \lambda_{i-1}$, and similarly define $Q_\lambda$.
% The \emph{$\gamma$-filtered bipartite neighborhood graph} is the union of the

The main way that ``geometry helps'' for matching problems in the plane is that one can use a geometric data structure to implicitly store this graph.
However, when working with approximations, one can show that the bipartite neighborhood graph has linear size and can be computed in linear time if one has already precomputed the filtered neighborhood graphs of the two sets.
Below, we explain the construction.

% Let $R_\lambda$ and $B_\lambda$ be the prefixes of the greedy permutations of $R$ and $B$ respectively that give $\lambda$-nets.

\begin{lemma}\label{lem:isolated_vertex_hausdorff_bound}
  Let $R,B\subset \R^2$.
  If $\binbrhd(R_\lambda, B_\lambda, \gamma\lambda)$ has an isolated vertex, then $\dist_H(R, B)\ge \lambda(\gamma-1)$.
\end{lemma}
\begin{proof}
  Without loss of generality, let $x\in R_\lambda$ be an isolated vertex.
  Then, there are no points of $B_\lambda$ within distance $\lambda\gamma$ of $x$ and so, $\dist_H(R, B_\lambda) \ge \lambda\gamma$.
  Because $B_\lambda$ is a $\lambda$-net, $\dist_H(B, B_\lambda)\le \lambda$.
  Therefore, by the triangle inequality, $\dist_H(R,B)\ge \lambda\gamma - \lambda = \lambda(\gamma-1)$.
\end{proof}

\begin{theorem}\label{thm:compute_binbrhd}
  Let $R, B$ be PDs and let $\lambda$ and $\gamma$ be constants such that $\lambda(\gamma-1)\ge \dist_H(R,B)$.
  Given the greedy permutations of $R$ and $B$ as well as the corresponding $(2\gamma+1)$-filtered neighborhood graphs, the $\binbrhd(R_\lambda, B_\lambda, \gamma\lambda)$ can be computed in linear time.
\end{theorem}
\begin{proof}
  The algorithm will be incremental, adding the points in order of their insertion radii.
  At each step $i$, we maintain $G_i = \binbrhd(R_{\lambda_i}, B_{\lambda_i}, \gamma\lambda_i)$, where $\lambda_i$ is the insertion radius of the newly inserted point.
  When inserting $p_i$, we add its neighbors $G_i$.
  Without loss of generality, assume $p_i\in R$.
  Let $y$ be the nearest neighbor of $p_i$ in $\nbrhd(R_i, 2\gamma(\lambda_i))$.
  So $\dist(y,p_i) = \lambda_i$.
  There are only a constant number of neighbors.
  Then, let $a$ be any neighbor of $y$ in $G_i$ (which are contained in the neighbors of $y$ in $G_{i-1}$).
  The neighbor $a$ must exist because of Lemma~\ref{lem:isolated_vertex_hausdorff_bound} and our choice of $\lambda$.
  By the triangle inequality,
  \[
    \dist(a,b)\le \dist(a,y) + \dist(y,x) + \dist(x, b) \le \lambda_i\gamma + \lambda_i + \lambda_i\gamma = \lambda_i(2\gamma+1).
  \]
  It follows that $a$ and $b$ are neighbors in $\nbrhd(B_{\lambda_i}, \lambda_i(2\gamma+1))$.
  So, the neighbors of $p_i$ can all be found among the neighbors of $y$.
  Because the degrees are constant and edges that are too long for the current value of $\lambda_i$ can be deleted as they are encountered the neighbors of $p_i$ can be found in amortized constant time.

  One pass over the edges suffices to remove any that are longer than $\lambda\gamma$.
  Thus, the total running time is linear.
\end{proof}

\begin{theorem}\label{thm:hausdorff_approximation}
Given the greedy permutations of $R$ and $B$ as well as the corresponding $(2\gamma+1)$-filtered neighborhood graphs, an approximation of $\dist_H(R,B)$ to within a factor of $1\pm \frac{1}{\gamma}$ can be computed in linear time.
\end{theorem}
\begin{proof}
  If the algorithm from Theorem~\ref{thm:compute_binbrhd} is run until it either inserts all the points or discovers an isolated vertex, it will produce a graph with a linear number of edges.
  By iterating over the edges, one can find, for each vertex, the distance to the nearest adjacent vertex in the graph.
  The maximum of these distances indicates the distance $r = \lambda\gamma$ at which a vertex becomes isolated and is the desired approximation.
  By Lemma~\ref{lem:isolated_vertex_hausdorff_bound}, we know that $\dist_H(R,B)\ge \lambda(\gamma-1) = r(1-\frac{1}{\gamma})$.
  Similarly, because $\dist_H(R, B_\lambda)\le \lambda$, the triangle inequality implies that $\dist_H(R,B) \le \dist_H(R, B\lambda) + \lambda = \lambda(\gamma + 1) = r(1 + \frac{1}{\gamma})$.
\end{proof}

%% file: conclusion.tex
\section{Conclusions and Future Work}
\label{sec:conclusion}

We have presented an efficient method to preprocess a persistence diagram so that one can extract guaranteed approximations with any number of distinct points.
It remains to explore the relationship between greedy PD sketches and other PD simplicifcations such as persistence landscapes~\cite{bubenik15statistical} or persistence images~\cite{adams17persistent}.
In particular, it is possible to construct a persistence landscape or a persistence image from a PD sketch, possibly much faster than with the entire diagram.
Another application where many PD distance computations are used is in the computation of Wasserstein barycenters.
In Vidal et al.~\cite{vidal19progressive}, a kind of sketch is used to dramatically speed up the Wasserstein barycenter computation.
In that case, they use the subset of points of highest persistence.
Although the approximation guarantees we prove are only applicable to the bottleneck distance, it seems reasonable that they should also be applicable to other Wasserstein metrics.
Indeed, we give the matching update for these metrics in Section~\ref{sec:matching_update}.

In future work, we will incorporate these sketches into a data structure that supports standard proximity search queries including (approximate) nearest neighbor search, metric range search, and metric range counting.
This is the main target of our work as it is a case where there is immediate benefit to finding fast approximate distances with bounds on the error to prune a search.

% TODO: How to cache partial computations?
% We should keep the matching.  This would help with NN search.  When moving down the tree, you can get a pretty good matching to start by composing matching computed from parent.
% An interesting open question is how to efficiently cache partial bottleneck computations.
% There are many cases where one can get away with an approximation at first, but will later need a more refined computation.

%% file: neighbor_analysis.tex
\section{Neighbor Analysis}
\label{sec:neighbor_analysis}

One of the key insights in the $sb$ data structure was the maintenance of a neighborhood graph with an edge between any pair of points $(a,b)$ for which the addition of a point the Voronoi cell of $a$ would impact the Voronoi cell of $b$.
Rather than maintaining this graph exactly, Clarkson showed that it sufficed to keep edges for which the distance is at most three times the larger of the radii of the cell.
With our modified distances and radii, we will need to extend this graph to four times the larger of the radii.
Below, we show that this is sufficient.

Note that for a diagonal point $\hat{a}$ and a point $x\in \vor(\hat{a})$, we have
\[
  \rad(\hat{a})\ge \max\{\|\hat{a}- \hat{x}\|, \|x- \hat{x}\|\}.
\]

Let $\hat{a}\in \hat{X}$ and $b\in X$ be inserted points.
Then, we need $\hat{a}$ and $b$ to be adjacent if there is $x\in \vor(\hat{a})$ and $y\in \vor(b)$ such that $\|x-y\|\le \max\{r(\hat{a},x), r(b,y)\}$.
Suppose this condition holds, then it follows that
\begin{align*}
  \|\hat{a}-b\|
    &\le \|\hat{a} - \hat{x}\| + \|x - \hat{x}\| + \|x-y\| + \|y-b\|\\
    &\le 2\rad(\hat{a}) + \max\{r(\hat{a},x), r(b,y)\} + \rad(b)\\
    &\le 2\rad(\hat{a}) + \rad(b) + \max\{\rad(\hat{a}), \rad(b)\}.
\end{align*}

For points $a\in X$ and $b\in X$, the situation is a little simpler.
Here we have $a$ and $b$ adjacent if there is $x\in \vor(a)$ and $y\in \vor(b)$ such that $\|x-y\|\le \max\{\|a-x\|, \|b-y\|\}$.
This is the case covered by Clarkson.
\begin{align*}
  \|a - b\|
    &\le \|a - x\| + \|x-y\| + \|y-b\|\\
    &\le \rad(a) + \rad(b) + \max\{\rad(a), \rad(b)\}.
\end{align*}

The last case to consider is when $\hat{a}, \hat{b}\in \hat{X}$.
In this case, we use the fact that $\|\hat{x} - \hat{y}\| \le \|x-y\|$.
\begin{align*}
  \|\hat{a} - \hat{b}\|
    &\le \|\hat{a} - \hat{x}\| + \|\hat{x}-\hat{y}\| + \|\hat{y}-\hat{b}\|\\
    &\le \|\hat{a} - \hat{x}\| + \|x-y\| + \|\hat{y}-\hat{b}\|\\
    &\le \rad(a) + \rad(b) + \max\{\rad(a), \rad(b)\}.
\end{align*}